\documentclass[citeauthoryear, 11pt]{llncs}
\usepackage{graphicx}
\usepackage {psfrag}
\usepackage{clrscode}

\begin{document}

\frontmatter

\pagestyle{headings}

\mainmatter

\title{Algorithms for Junctions\\ in Directed Acyclic Graphs}

\author{Carlos Eduardo Ferreira\thanks{The first author has been supported by CNPq (proc. no. 302736/2010-7)} \and \'Alvaro Junio Pereira Franco\thanks{The second author has been supported by CAPES (proc. no. 33002010176P0)}}

\institute {Universidade de S\~ao Paulo, Instituto de Matem\'atica e Estat\'istica, \\Rua do Mat\~ao, Cidade Universit\'aria, CEP 05508-090, S\~ao Paulo, Brasil \email{\{cef,alvaro\}@ime.usp.br}}

\maketitle

\begin{abstract}
Given a pair of distinct vertices $u, v$ in a graph $G$, we say that $s$ is a junction of $u, v$ if there are in $G$ internally vertex disjoint directed paths from $s$ to $u$ and from $s$ to $v$. We show how to characterize junctions in directed acyclic graphs. We also consider the two problems in the following and derive efficient algorithms to solve them. Given a directed acyclic graph $G$ and a vertex $s \in G$, how can we find all pairs of vertices of $G$ such that $s$ is a junction of them? And given a directed acyclic graph $G$ and $k$ pairs of vertices of $G$, how can we preprocess $G$ such that all junctions of $k$ given pairs of vertices could be listed quickly? All junctions of $k$ pairs problem arises in an application in Anthropology and we apply our algorithm to find such junctions on kinship networks of some brazilian indian ethnic groups.
\end{abstract}
\section{Introduction}
Given a directed graph $G$, a vertex $s\in G$ is a junction of a pair $u, v \in G$, $u \neq v$, if there exist two internally disjoint directed paths from $s$ to $u$ and from $s$ to $v$. A vertex $s \in G$ is a lowest common ancestor, LCA for short, of a pair of vertices $u, v$, $u\neq v$, if it is a junction and for every junction $s^\prime \neq s$ of the pair $u, v$, there is no directed path in $G$ from $s$ to $s^\prime$. In rooted trees, junctions and LCAs are the same vertices and given a pair $u,v \in G$, $u\neq v$, the LCA is unique. On the other hand, in directed graphs, acyclic or not, we can have several different LCAs or junctions for a certain pair of vertices. The problem of finding out whether $s$ is a junction, or an LCA of a pair $u, v$ can be done efficiently using a maximum-flow algorithm. So, we can determine whether a given vertex is a junction, or an LCA, on directed graphs in polynomial time. There are faster algorithms that find out junctions and LCAs in rooted trees and directed acyclic graphs. 

\noindent{\it Related Works.} A problem that have been studied in rooted trees and directed acyclic graphs (DAGs) is how could we preprocess a given graph such that a query to a representative LCA for any pair of vertices could be done quickly. We have found in the literature fast algorithms that solve this problem. Aho, Hopcroft and Ullman (\cite{ahu}) have shown how to preprocess rooted trees with $n$ vertices in $O(n\alpha(n))$ time, where $\alpha(n)$ is the number of times that we must apply $\log_2$ to $n$ to obtain a number less than or equal to zero. Queries are answered in constant time. Harel and Tarjan (\cite{ht}) have shown how to preprocess in $O(n)$ time and to answer queries on the representative LCA in constant time. Other algorithms to preprocess rooted trees appeared in the works by Berkman and Vishkin (\cite{bv2}), Nyk\"anen and Ukkonen (\cite{nu}) and Wen (\cite{w}). The equivalent problem to DAGs is known as {\sc all-pairs-lca} (Bender, Farach-Colton, Pemmasani, Skiena and Sumazin \cite{bfpss}). Given a DAG with $n$ vertices, the preprocessing phase proposed by Bender, Farach-Colton, Pemmasani, Skiena and Sumazin (\cite{bfpss}) spends $\tilde{O}(n^{2.688})$\footnote{$f(n) = \tilde{O}(g(n))$ if there is a constant $c$ such that $f(n) = O(g(n)\log^c n)$.} and a query to a representative LCA can be answered in constant time. Later, Czumaj, Kowaluk and Lingas (\cite{ckl}) have solved this problem to DAGs with preprocessing time $O(n^{2.575})$ and Eckhardt, M\"uhling and Nowak (\cite{emn}) with expected time $O(n^2\log n)$.

In the problem {\sc all-pairs-all-lcas} we want to know how to preprocess a given DAG with $n$ vertices and $m$ arcs such that a query to all LCAs of a pair of vertices can be done quickly. Baumgart, Eckhardt, Griebsch, Kosub and Nowak (\cite{begkn}) have developed algorithms that solve {\sc all-pairs-all-lcas}. The upper bound of the preprocessing is $O(\min\{n^2m, n^{3.575}\})$. Eckhardt, M\"uhling and Nowak (\cite{emn}) have developed two algorithms for this problem, one with preprocessing expected time $O(n^3\log\log n)$ and the other with preprocessing time $O(n^{3.3399})$.

Yuster (\cite{y}) considers in his work the {\sc all-pairs-junction} problem, i.e., how to preprocess a DAG with $n$ vertices such that a query to a representative junction can be done quickly? Any algorithm that solves {\sc all-pairs-lca} could be used to solve {\sc all-pairs-junction} but Yuster (\cite{y}) have shown that we can preprocess in $\tilde{O}(n^\omega)$ time and answer queries in constant time, where $w < 2.376$ is the exponent of fast boolean matrix multiplication.

In the problem treated by Tholey (\cite{t}) there is given a DAG $G$, $s_1, s_2 \in G$ and $k$ pairs of vertices $(u_1, v_1), \dots, (u_k, v_k)$. For each pair $(u_i, v_i)$, $1 \leq i \leq k$, a tuple $(s_1,t_1, s_2, t_2)$ with $\{t_1, t_2\} = \{u_i, v_i\}$ is printed out in constant time if there are two disjoint directed paths from $s_1$ to $t_1$ and $s_2$ to $t_2$. Before that he constructs a modified version of the data structure proposed by Suurballe and Tarjan (\cite {st}) in $O(n\log^2 n + (m+k)\log_{2+(m+k)/(n+k)} n)$ time. 

\noindent{\it Our Application.} This work has been motivated by an application that arises on the field of Anthropology. It is given a kinship network (a DAG) that models the parent-child relationships from a determined society (in our case indian Brazil ethnic groups). Moreover, the set of weddings among individuals from that society is given. We want to know if the partners of the wedding are relatives in some degree. From the anthropological point of view, the junctions, i.e., common ancestors with disjoint descendant lines, are the objects of interest (dal Poz and Silva \cite{dps}). In this case, we want to find all junctions of many pairs of individuals of the network. Thus, we treated here the problem {\sc $k$-pairs-all-junctions}, that is, given a DAG $G$ and $k$ pairs of vertices of $G$, how can we preprocess $G$ such that a query to all junctions of $k$ given pairs of vertices can be done quickly?

\noindent{\it Main Results.} We develop an algorithm to solve the following problem named \textsc{single-junction-all-pairs}. Given a DAG $G$ with $n$ vertices and $m$ arcs, and a vertex $s \in G$, construct a data structure that allows to find all pairs of vertices for which $s$ is a junction. The time spent by it is $O(m)$. We can use another data structure  named dominator trees (Aho and Ullman, \cite{au}) to solve \textsc{single-junction-all-pairs} problem. We construct a dominator tree $T$ rooted in $s$ in linear time considering the graph induced by the vertices descendants of $s$ in $G$. The pairs $u, v$ that have $s$ as a lowest common ancestor in $T$ are the pairs that have $s$ as a junction in $G$. However, our data structures is simpler than the data structures for dynamic LCA queries in trees (Cole and Hariharan, \cite{ch}) used to construct a dominator tree in DAGs. 

To solve the \textsc{k-pairs-all-junctions} problem, we just solve the problem \textsc{single-junction-all-pairs} for all $s$ in $G$ listing (or storing) $s$ to pair $u,v$, if $s$ is a junction of $u,v$. The time spend by it is $O(n(m+k))$.

The algorithm to $\textsc{k-pairs-all-junctions}$ problem can be used to solve the $\textsc{k-pairs-all-lcas}$ problem. Given a DAG $G$ with $n$ vertices and $m$ arcs, and $k$ pairs of vertices, we first find the transitive closure of $G$, find all junctions of the $k$ given pairs of vertices, and finally for each given pair $u, v$ and for each junction $s_i \in J(u,v)$, we verify if there exists an arc $s_i-s_j$ in the transitive closure of $G$, where $s_j \neq s_i$ and $s_j \in J(u,v)$. If yes, then $s_i$ cannot be an LCA of the pair $u,v$ because $s_i$ has $s_j$ as a common descendant to $u, v$. If no, then $s_i$ is an LCA of the pair $u,v$. Thus, the time spent by this method is $O(n^\omega + n(m + k) + n^2k)$, where $O(n^\omega)$ is the time of transitive closure of $G$. If $k=o(n^{1.3399})$, then our simple approach for \textsc{k-pairs-all-lcas} is faster than the best algorithm (worst case) known for solving \textsc{all-pairs-all-lcas} problem. 


Any algorithm that aims to list all junctions of all pairs of vertices spends $\Omega(n^3)$ time, in the worst case. To see that consider the example in Fig. \ref{fig:lim_inf}, where each vertex in the first line is a junction of every pair of vertices in the second one. Thus, to list all the junctions of all $\Omega(n^2)$ pairs would spend $\Omega(n^3)$ time.
\begin{figure}[htb]
\label{fig:lim_inf}
\begin{center}
\psfrag{O(n)}{\small $\Omega(n)$}
\psfrag{...}{\small $\dots$}
\includegraphics[scale=.45]{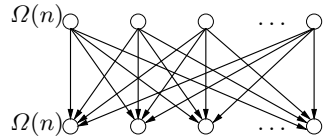}
\end{center}
\caption{Pairs with $\Omega(n)$ junctions.}
\end{figure}

\noindent{\it Organization.} This text is organized as follows. In Section $2$, we introduce some notation used in this work, recall an ordinary search algorithm (depth-first search) and discuss some simple data structures. In Section $3$, we characterize junctions in DAGs. In Section $4$, we describe the algorithm for \textsc{single-junction-all-pairs} problem, and in Section $5$, we make our last considerations. 

\section{Initial Concepts}
Given a DAG $G$ and vertices $s, t \in G$ we say $s$ is a parent of $t$ or $t$ is a child of $s$, when there exists an arc $s-t \in G$. We denote the set of parents of a vertex $u$ in $G$ by $\delta_G^-(u)$ and the set of children of a vertex $u$ in $G$ by $\delta_G^+(u)$. We say $s$ is an ancestor of $t$ or $t$ is a descendant of $s$ when there exists in $G$ a directed path from $s$ to $t$. Usually, we treat a directed path $P=\{s=t_0, t_1, \dots, t_{k-1}=t\}$ as an ordered set of vertices with the property that $t_i - t_{i+1} \in G$, for all $i = 0, \dots, k-2$.

If a vertex $s$ is a junction of the pair of distinct vertices $u, v$, then there are directed paths internally disjoint $P$ from $s$ to $u$ and $Q$ from $s$ to $v$. Note that, if $s$ is an LCA of $u, v$, then $s$ is a junction of $u, v$, but the converse is not always true. As we mentioned before, a pair of vertices $u, v \in G$ can have many LCAs and junctions. We write $J(u,v)$ to denote the set of vertices of $G$ that are junctions of the pair $u, v$. The sets $J(u, v)$ and $J(v, u)$ are equal. If $u$ is a proper ancestor of $v$, then $u \in J(u, v)$. We suppose that $J(u, u) = \emptyset$, for all $u\in G$. Analogously, $LCA(u,v)$ denotes the set of LCAs of the pair $u, v$. If we want to specify the graph $G$ we are working on, then we write $LCA_G(u,v)$. When $G$ is a rooted tree, we abuse of the language writing $s = LCA_G(u,v)$. 

We call an arborescence a rooted directed tree, such that: 1 - there exists a unique vertex with in-degree equal to zero called root; 2 - all vertices except the root have in-degree one; and 3 - there exists a directed path from the root to all vertices in the arborescence. An arborescence with root $s$ is denoted by $T_s$. Any vertex $u \in T_s$ is the root of a subarborescence denoted by $T_u \subseteq T_s$. 

Given a DAG $G$ and a vertex $s \in G$, we use a depth-first search to construct an arborecence $T_s$ composed by root $s$ and all descendants of $s$ in $G$. We denote by $F_s$ the set of all arcs of $G$ with both ends in $T_s$ that are not arcs of $T_s$. For each vertex $u\in T_s$, we maintain the integer post-order value stored in array $post$ indexed by vertices. We also maintain the integer values $minpost[u]$, that is, the minimum post-order values in $T_u$. After such construction, we can make a partition of the arcs with both ends in $T_s$: 1. {\sl Arcs of the arborescence $T_s$}; 2. {\sl External descendant arcs} -- are the arcs in $F_s$ of the form $s-v$, $v \in T_s$. Note that $v$ is a child of $s$ in $G$ but is not a child in $T_s$; 3. {\sl Internal descendant arcs} -- are the arcs $u-v \in F_s$, $u\neq s$, such that there exists a directed path in $T_s$ from $u$ to $v$; 4. {\sl External crossing arcs} -- are the arcs $u-v \in F_s$ such that $u \in T_{s_2}$ and $v \in T_{s_1}$, where $s_1$ and $s_2$ are two different children of $s$ in $T_s$; 5. {\sl Internal crossing arcs} -- are the arcs $u-v \in F_s$ such that $u \in T_{s_1}$, $v \in T_{s_1}$, $s_1$ is a child of $s$ in $T_s$ and there are no directed paths in $T_s$ from $u$ to $v$ neither from $v$ to $u$. 
A similar partition of the arcs has been done by Sedgewick (\cite{s}) and Dasgupta, Papadimitriou and Vazirani (\cite{dpv}). Such construction has some properties. 
\begin{property}
\label{p:1}
For any arc $u-v$ with both ends in $T_s$ we have $post[u] > post[v]$; 
If $s_1$ and $s_2$ are two children of $s$ in $T_s$ and $post[s_1] < post[s_2]$ then,
there is no directed path from $p$ to $q$, where $p \in T_{s_1}$ and $q \in T_{s_2}$; and $post[p] < post[q]$, for all $p \in T_{s_1}$ and for all $q \in T_{s_2}$.
\end{property}
\begin{proposition}
\label{prop:1}
Let $G$ be a DAG, $s$ a vertex in $G$ and $T_s$ an arborescence constructed by depth-first search as before described. Let $s_1$ be a child of $s$ in $T_s$ and $u$ be a vertex in $T_{s_1}$. Let $P=\{s=u_0, u_1, \dots, u_{k-1}=u\}$ be a directed path from $s$ to $u$ in $G$. If the vertex $u_i$, for some $i = 1, \dots, k-1$, is the first vertex of $P$ in $T_{s_1}$, then all remaining vertices of $P$, $u_{i+1}, \dots, u_{k-1}$, belong to $T_{s_1}$. \qed
\end{proposition}
\section{Some Results for Junctions in DAGs}
In this section, let us consider a DAG $G$, a vertex $s \in G$ and an arborescence $T_s$ constructed as mentioned before. The arcs of $G$ that do not belong to $T_s$ or $F_s$ are not used in any directed path beginning in $s$. So, they are not important to obtain the pairs $u, v$ that have $s$ as a junction.
The following proposition can be quickly checked.
\begin{proposition}
\label{prop:2}
Let $T_s$ be an arborescence with root in $s$, $s_1$ and $s_2$ different children of $s$ in $T_s$ and two vertices $u, v \in T_s$. If $u \in T_{s_1}$ and $v \in T_{s_2}$, then $s \in J(u,v)$ and $s=LCA_{T_s}(u,v)$. \qed
\end{proposition}

It remains to us to find out whether $s$ is a junction of pairs $u, v$ when $u$ and $v$ belong to a same subarborescence. Note that if we take two vertices $u$ and $v$ that belong to subarborescence $T_{s_i}$, $s_i$ child of $s$ in $T_s$, and if there is no external arc $p-q$ with $q \in T_{s_i}$, then $s$ cannot be a junction of this pair, since all directed path from $s$ to $u$ and from $s$ to $v$ necessarily share the vertex $s_i$. Therefore, if $s \in J(u, v)$ with $u, v \in T_{s_i}$, then there exists a pair of internally disjoint paths, from $s$ to $u$ and from $s$ to $v$, and one of them has an external arc.

Additionally, we note that for any vertex $s_1$ child of $s$ in $T_s$ and a pair $u, v$ in $T_{s_1}$, if $s \in J(u,v)$, then there exists a pair of internally disjoint paths in $G$ from $s$ to $u$ and from $s$ to $v$, such that all internal vertices from one of them are in $T_{s_1}$. We can prove this by contradiction. Consider two internally disjoint paths in $G$, $P = \{s=u_0, \dots,$ $u_{k-1} = u\}$ and $Q=\{s=v_0, \dots, v_{l-1}=v\}$. Suppose that $P$ and $Q$ do not have all vertices in $T_{s_1}$. Consider the first vertices $u_i$ and $v_j$ of $P$ and $Q$, respectively, that belong to $T_{s_1}$. By Proposition \ref{prop:1}, the directed paths $P_1=\{u_i, u_{i+1}, \dots, u_{k-1}\}$ and $Q_1=\{v_j, v_{j+1}, \dots, v_{l-1}\}$ are internal to $T_{s_1}$. Take the directed path $P_2$ from $s_1$ to $u_i$ in $T_{s_1}$. If $P_2 \cap Q_1 = \emptyset$, then we can construct a directed path from $s$ to $u$ $(\{s\} \cup P_2 \cup P_1)$ with all vertices in $T_{s_1}$, where $(\{s\}\cup P_2 \cup P_1) \cap Q = \{s\}$. So, $P_2 \cap Q_1 \neq \emptyset$. Call $v^\prime$ the vertex from this intersection nearest to $v_{l-1}$. In the same way, take the directed path $Q_2$ from $s_1$ to $v_j$ in $T_{s_1}$. If $Q_2 \cap P_1 = \emptyset$, then we can construct a directed path from $s$ to $v$ $(\{s\}\cup Q_2 \cup Q_1)$ with all vertices in $T_{s_1}$ and $(\{s\}\cup Q_2 \cup Q_1) \cap P = \{s\}$. So, $Q_2 \cap P_1 \neq \emptyset$. Call $u^\prime$ the vertex from this intersection nearest to $u_{k-1}$. Thus, we find a cycle with the directed paths from $v_j$ to $v^\prime$, from $v^\prime$ to $u_i$, from $u_i$ to $u^\prime$ and from $u^\prime$ to $v_j$. A contradiction, since $G$ is a DAG. Therefore, all internal vertices to $P$ or $Q$ are in $T_{s_1}$. This fact helps us to prove the following lemma.

\begin{lemma}
\label{l:1}
Let $s_1$ be a child of $s$ in $T_s$, $u$ and $v$ belong to $T_{s_1}$ and $z = LCA_{T_s}(u, v)$ with $z \neq u, v$. The vertex $s \in J(u, v)$ if, and only if, $z$ belongs to a pair of internally disjoint paths from $s$ to $u$ or from $s$ to $v$.
\end{lemma}
\begin{proof}
Take a pair of internally disjoint paths $P=\{s = u_0, \dots, u_{k-1} = u\}$ and $Q = \{s = v_0, \dots, v_{l-1} = v\}$. Suppose all internal vertices of $P$ are in $T_{s_1}$ and $Q$ has an external arc $v_i-v_{i+1}$ entering in $T_{s_1}$. Consider $R = \{s_1 = z_0, \dots, z_{m-1} = z\}$ to be the directed path in $T_{s_1}$ from $s_1$ to $z$. Consider $P^\prime = \{z = u^\prime_0, \dots, u^\prime_{n-1} = u\}$ the directed path in $T_{s_1}$ from $z$ to $u$ and $Q^\prime = \{z = v^\prime_0, \dots, v^\prime_{r-1} = v\}$ the directed path in $T_{s_1}$ from $z$ to $v$. Here is really important to note that, by the construction of $T_s$, any path that passes through a vertex in $Q^\prime$ $(P^\prime)$, and then passes through a vertex in $P^\prime$ $(Q^\prime)$, cannot come back to $Q^\prime$ $(P^\prime)$.  Let us divide the proof in two cases.\\
{\bf Case 1.} There is no vertex of $Q$ in $R$. If there is no vertex of $Q$ in $P^\prime$, then we can make $Z=\{s\} \cup R \cup P^\prime$ and we have $Z$ and $Q$ internally disjoint. If there is some vertex of $Q$ in $P^\prime$, then consider $v_j = u^\prime_i$ the vertex in $Q \cap P^\prime$ nearest to $z$. So, the directed paths $U = \{s = v_0, \dots, v_j\} \cup \{v_j = u^\prime_i, \dots u^\prime_{n-1} = u\}$ and $Z = \{s\} \cup R \cup Q^\prime$  are internally disjoint. See the two first illustrations in Fig. \ref{fig:4}.\\
{\bf Case 2.} There are vertices of $Q$ in $R$. Let us consider that the vertex nearest to $z$ in $(P \cup Q) \cap R$ belongs to $Q$, i.e., $v_w=z_t$. The case when this vertex belongs to $P$ is symmetric.\\
{\bf Case 2.1.} There is no vertex of $P$ in $Q^\prime$. Then, the following directed paths are internally disjoint: $P$ and $Z = \{s=v_0, \dots, v_w\} \cup \{v_w = z_t, \dots, z_{m-1}=z\} \cup Q^\prime$.\\
{\bf Case 2.2.} There are some vertices of $P$ in $Q^\prime$. Consider $u_p=v^\prime_q$ the vertex in $P \cap Q^\prime$ nearest to $z$. So, the following directed paths are internally disjoint: $Z=\{s=v_0, \dots, v_w\} \cup \{v_w=z_t, \dots, z_{m-1} = z\} \cup P^\prime$ and $U=\{s = u_0, \dots, u_p\} \cup \{u_p = v^\prime_q, \dots, v^\prime_{r-1} = v\}$.

        \begin {figure}[htb]
          \begin {center}

            \psfrag {vj}{\tiny $v_j$}
            
            \psfrag {s}{\tiny $s$}
            \psfrag {s1}{\tiny $s_1$}
            \psfrag {Z}{\tiny $Z$}
            \psfrag {P}{\tiny $P$}
            \psfrag {U}{\tiny $U$}
            \psfrag {Q}{\tiny $Q$}
            \psfrag {u}{\tiny $u$}
            \psfrag {v}{\tiny $v$}
            \psfrag {z}{\tiny $z$}
            \psfrag {uz}{\tiny $u_z$} 
            \psfrag {vz}{\tiny $v_w$}
            \psfrag {vq}{\tiny $v_q$}
            \psfrag {up}{\tiny $u_p$}
            \includegraphics[scale=.45]{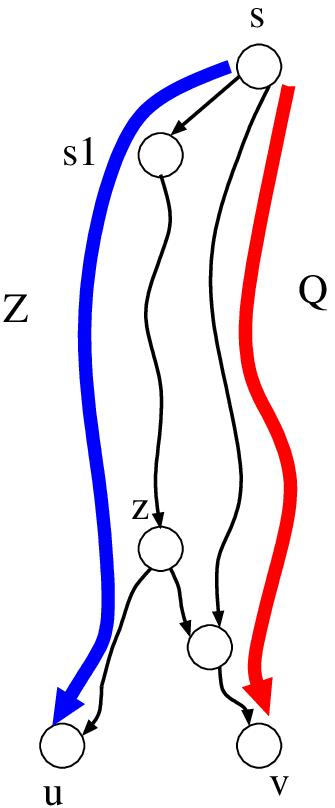}
            \hspace{1cm}
            \includegraphics[scale=.45]{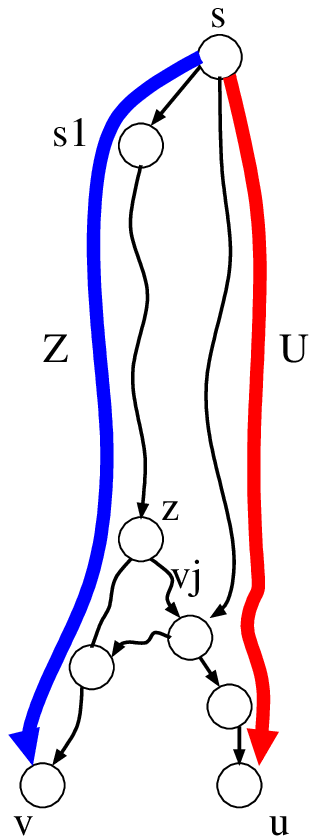}
            \hspace{1cm}
            \includegraphics[scale=.45]{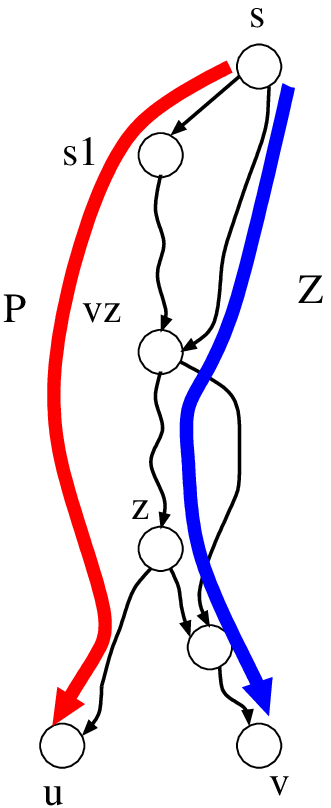}
            \hspace{1cm}
            \includegraphics[scale=.45]{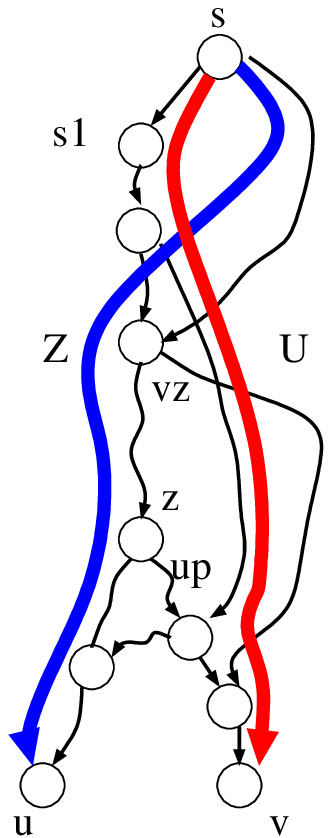}
            \hspace{1cm}
          \end {center}
          \caption {New internally disjoint paths Lemma \ref{l:1} -- Cases 1 and 2. One of them includes vertex $z = LCA_{T_s}(u,v)$. The vertex $v_w \in Q$ is the nearest to $z$ in $(P \cup Q) \cap R$. The unlabeled vertices are not important. They are in the figure just to indicate the input vertices in a path of $T_s$.}
          \label{fig:4}
          \end {figure}
Therefore, if $s \in J(u, v)$, then there exists a pair of internally disjoint paths with $z$ belonging to one of them. This ends the first part of the proof.\\
The converse is easier. By definition, $s$ is a junction of the pair $u, v$ if there is a pair of internally disjoint paths, independently if $z$ belongs to one of them.\qed
\end{proof}
The Lemma \ref{l:2} is a consequence of Lemma \ref{l:1}. It characterizes $s$ as a junction of a pair $u, v$ when we have $z = LCA_{T_s}(u,v)$ and $z \neq u, v$. Lemma \ref{l:3} characterizes $s$ as a junction of a pair $u, v$ in the case where $z=u$ or $z=v$.
\begin{lemma}
\label{l:2}
Let $s_1$ be a child of $s$ in $T_s$, $u$ and $v$ belong to $T_{s_1}$, $z = LCA_{T_s}(u, v)$ and $z \neq u, v$. The vertex $s \in J(u, v)$ if, and only if, $s \in J(z, u)$ or $s \in J(z, v).$\qed
\end{lemma}
\begin{lemma}
\label{l:3}
Let $s_1$ be a child of $s$ in $T_s$, $z$ and $u$ belong to $T_{s_1}$ and $z$ be a proper ancestor of $u$ in $T_s$, i.e., $z = LCA_{T_s}(z, u)$, $z\neq u$. The vertex $s \in J(z, u)$ if, and only if, $s \in J(z, t)$, for some $t \in \delta_G^-(u)$.
\end{lemma}
\begin{proof}
To prove the first part we take two internally disjoint paths $P=\{s = u_0, \dots, u_{k-1} = u\}$ and $Q=\{s = z_0, \dots, z_{l-1} = z\}$. We know that $u_{k-2} \in \delta_G^-(u)$ since it is a parent of $u$ in $G$. Therefore, the directed paths $U = \{s = u_0, \dots, u_{k-2}\}$ and $Q$ are disjoint, and then $s \in J(z, u_{k-2})$. \\
To prove the converse, we consider two internally disjoint paths $P=\{s = t_0, \dots, t_{k-1} = t\}$ and $Q=\{s = z_0, \dots, z_{l-1} = z\}$. If $u = t_i$, for any $i = 1, \dots, k-2$, then the directed path $\{u = t_i, t_{i+1}, \dots, t_{k-1}, u\}$ is a cycle. But $G$ is a DAG. Thus, $u\notin P$. If $u = z_j$, for any $j = 1, \dots, l-2$, then $u$ is a proper ancestor of $z$. By assumption, $z$ is a proper ancestor of $u$. Thus, we can again produce a cycle contradicting  the fact that $G$ is a DAG. Then, $u \notin Q$. Therefore, we can use the arc $t-u$ extending the directed path $P$ and constructing two internally disjoint paths $U=P \cup \{u\}$ and $Q$. So, $s \in J(z, u)$. \qed
\end{proof}
The next two lemmas help us to detect, respectively, some pairs of vertices that have and some that do not have $s$ as a junction. 
\begin{lemma}
\label{l:4}
Let $s_1$ be a child of $s$ in $T_s$, $z$ and $w$ belong to $T_{s_1}$ and $z$ be a proper ancestor of $w$ in $T_s$, i.e., $z=LCA_{T_s}(z, w)$, $z \neq w$. Consider the directed path from $z$ to $w$ in $T_s$, $Z=\{z = w^\prime_0, \dots, w^\prime_{k-1}=w\}$. If the vertex $s \notin J(z, w^\prime_i)$, for all $i = 1, \dots, k-2$ and $s \in J(z, w)$, then 
\begin {itemize}
\item [{\bf a.}] the vertex $s \in J(w^\prime_i, w)$, for all $i = 1, \dots, k-2$; and
\item [{\bf b.}] for all pairs of vertices $u, v$, $u \in T_z\setminus T_w$ and $v \in T_w$, we have $s \in J(u, v)$.
\end {itemize}
\end{lemma}
\begin{proof}
{\bf a.} Consider two internally disjoint paths $P=\{s = w_0, \dots, w_{l-1} = w\}$ and $Q=\{s = z_0, \dots, z_{m-1} = z\}$. We know that $(Q \cap Z) \setminus \{z\} = \emptyset$, since $Z \setminus \{z\}$ contains proper descendants of $z$ and $Q \setminus \{z\}$ contains proper ancestors of $z$. Note also that $(P \cap Z) \setminus \{w\} = \emptyset$. If it is not the case, then $s \in J(z, w^\prime_i)$ for some $i = 1, \dots, k-2$, against our assumption. So, we can extend the directed path $Q$ using the path $Z$ from $z=w^\prime_0$ until $w^\prime_i$, for all $i = 1, \dots, k-2$. Therefore, we can construct new directed paths $R_i=Q \cup \{z = w^\prime_0, \dots, w^\prime_i\}$, for all $i = 1, \dots, k-2$, such that the paths $R_i$ and $P$ are internally disjoint and $s \in J(w^\prime_i, w)$, for all $i = 1, \dots, k-2$. \\
{\bf b.} Take $u \in T_z\setminus T_w$ and $v \in T_w$. Let $w^\prime_j = LCA_{T_s} (u,v)$. By Lemma \ref{l:4} a, we have, in particular, $s\in J(w^\prime_j, w)$. Therefore, there exist directed paths internally disjoint $P$ from $s$ to $w^\prime_j$ and $Q$ from $s$ to $w$. Note that we can extend the directed path $Q$ adding the path in $T_w$ from $w$ to $v$ and this new path, denoted by $R$, does not intersect $P$ or $Q\setminus \{w\}$. If it does, there would be a vertex ancestor and descendant of $w$. Thus, $R \cap P = \{s\}$. Therefore, we have $s \in J(w^\prime_j, v)$. If $u=w^\prime_j$, then is done. If $u \neq w^\prime_j$, then by Lemma \ref{l:2}, $s \in J(u,v)$. \qed
\end{proof}
\begin{lemma}
\label{l:5}
Let $s_1$ be a child of $s$ in $T_s$, $z$ and $w$ belong to $T_{s_1}$ and $z$ be a proper ancestor of $w$ in $T_s$, i.e., $z=LCA_{T_s}(z, w)$, $z \neq w$. Consider the directed path from $z$ to $w$ in $T_s$, $Z=\{z = w^\prime_0, \dots, w^\prime_{k-1}=w\}$. If the vertex $s \notin J(z, w^\prime_i)$ for all $i = 1, \dots, k-1$, then $s \notin J(w^\prime_i, w^\prime_j)$, for all $i,j = 1, \dots, k-1$.
\end{lemma}
\begin {proof}
Proof by contradiction. Suppose that for some $i, j = 1, \dots, k-1$, we have $s \in J(w^\prime_i, w^\prime_j)$. Consider $P = \{s=x_0, \dots, x_{k-1} = w^\prime_i\}$ and $Q = \{s=y_0, \dots, y_{l-1}=w^\prime_j\}$. Consider the directed path on arborescence $T_s$, $R = \{s=z^\prime_0, \dots, z^\prime_{t-1} = z\}$. Define $W$ as the set of vertices of $R \cap (P \cup Q)$. Take the vertex $z^\prime_q \in W$ nearest from $z$. Suppose that $z^\prime_q = y_r \in Q$ (the other case, $z^\prime_q \in P$ is analogous). Then, the directed paths $Z = \{s = y_0, \dots, y_r = z^\prime_q\} \cup \{z^\prime_q, z^\prime_{q+1}, \dots, z^\prime_{t-1} = z\}$ and $P$ are internally disjoint. Therefore, $s \in J(z, w^\prime_i)$, a contradiction. Thus, for all $i, j = 1, \dots, k-1$, we have $s \notin J(w^\prime_i, w^\prime_j)$. \qed
\end{proof}

\section{Algorithms for Junctions in DAGs}
Let us develop an algorithm that solves the problem where it is given a DAG $G$ and a vertex $s \in G$ and we want to know which pairs of vertices $u,v \in G$, $u\neq v$ have $s$ as a junction. We use a data structure for disjoint sets represented by rooted trees and we apply the path compression heuristic. In this data structure, each set is identified by an element called representative (chosen as the root of a tree). See Cormen, Leiserson, Rivest and Stein (\cite{clrs}), for more details. We use an array $p$ indexed by vertices to represent such structure. For each vertex $v$, $p[v]$ points to its respective parent in the data structure for disjoint sets. If $p[v]=v$, then $v$ is a representative. 
The $\proc{Single-Junction-All-Pairs}$ algorithm receives two vertices $s, z \in T_s$ and constructs a data structure for disjoint sets, represented by the global array $p$, such that $s$ is a junction of a pair $u, v \in T_z$ if, and only if, $p[u] \neq p[v]$. 
\begin{codebox}
\Procname{$\proc{Single-Junction-All-Pairs}(s, z)$}
\li \If $|T_z| > 1$
\li \Then $w \gets \proc{vertex}\ (post[z] - 1)$
\li       \While $post[w] \geq minpost[z]$
\li       \Do \If $\exists\ t \in \delta_G^-(w)$ with $p[z] \neq  p[t]$
\li           \Then $p[w] \gets w$ 
\li                 $\proc{Single-Junction-All-Pairs}(s, w)$
\li                 $w \gets  \proc{vertex} (minpost[w] - 1)$
\li            \Else $p[w] \gets z$
\li                 $w \gets \proc{vertex} (post[w] - 1)$
              \End
          \End
     \End
\end{codebox}

The function $\textsc{vertex} (p)$ receives an integer $p$ and returns a vertex $v$ such that $post[v] = p$, if $post[v]$ is defined; or a dummy vertex with post-order value equal to $-1$, otherwise. It spends constant time. 

To construct a data structure for disjoint set such that $s$ is a junction of a pair $u, v \in T_s$ if, and only if, $p[u] \neq p[v]$, we have to make the calls $\proc{Single-Junction-All-Pairs} (s, s_i)$, for all $s_i \in \delta^+_{T_s}(s)$.
We assume that $T_s$ has been constructed as before described and the array $p$ has been properly initialized as follows. For all arcs $u-v \in T_s$, $p[v] = u$, $p[s] = s$ and $p[s_i] = s_i$, for all $s_i \in \delta^+_{T_s}(s)$. Initially we have $|\delta^+_{T_s}(s)| + 1$ disjoint sets and for now $u, v$ with $u \in T_{s_i}$ and $v \in T_{s_j}$, $s_i$ and $s_j$ different children of $s$ in $T_s$ are pairs of vertices for which $s$ is a junction, as we have said in Proposition \ref{prop:2}. Let us see that the algorithm $\proc{Single-Junction-All-Pairs}$ is correct. 
\begin {lemma}
\label{l:7}
Suppose a call to ${\proc{Single-Junction-All-Pairs} (s, s_i)}$, where $s_i \in \delta_{T_s}^+(s)$. In line 3, the following assertions are true for the parameter vertex $z$ and for all pairs of vertices $u, v \in T_{s_i}$, such that $post[u], post[v] \geq post[w]+1$. 
\begin {itemize}
\item The vertex $z$ is a representative;
\item The values $p[u]$ and $p[v]$ point to a representative;
\item For all vertices $x \in T_{s_i} \setminus T_z$, $p[z] \neq p[x]$; and
\item  The vertex $s \in J(u,v)$ if, and only if, $p[u] \neq p[v]$. \hspace{2.9cm} $\qed$
\end{itemize}
\end{lemma}
Note that each arc in $G$ is explored once by $\proc{Single-Junction-All-Pairs}$ algorithm. So, the next result follows.   
\begin{theorem}
\label{t:1}
Given a DAG $G$ with $n$ vertices and $m$ arcs, and a vertex $s \in G$, we can construct a data structure for disjoint sets in time $O(m)$, and answer in constant time for any pair of vertices $u, v\in G$, if $s$ is or is not a junction of $u,v$. Moreover, after of such construction, all $k$ pairs of vertices that have $s$ as a junction can be printed out in $O(k)$ time. \qed
\end{theorem}
Theorem \ref{t:1} solves the \textsc{single-junction-all-pairs} problem and helps us to solve other problem associated.
\begin{corollary}
\label{c:1}
Given a DAG $G$ with $n$ vertices and $m$ arcs, a vertex $s \in G$, and $k$ pairs of vertices of $G$, $(u_1, v_1), \dots, (u_k, v_k)$, we can answer in $O(m+k)$ time which pairs $(u_i, v_i)$, $1 \leq i \leq k$, have $s$ as a junction. \qed
\end{corollary}
Applying Corollary \ref{c:1} for each $s \in G$,  we solve the {\sc $k$-pairs-all-junctions} in $O(n(m + k))$ time. We spend, respectively, $O(m)$ or $O(nk + m)$ space, if we print out during the execution or we store all junctions of the $k$ given pairs. 

\section{Conclusion}
Given a DAG $G$ with $n$ vertices and $m$ arcs, and $k$ pairs of vertices of $G$, we have interest in listing all junctions of the given pairs. This problem, called {\sc $k$-pairs-all-junctions}, arises in some application in Anthropology. It is solved in two parts. First, we construct a data structure for disjoint sets in time $O(m)$ such that a pair of vertices $u,v$ in $G$ has $s$ as a junction if, and only if, $u$ and $v$ are in different sets. So, from $k$ given pairs of vertices, we find out which of them have $s$ as a junction in $O(k)$ time since we can say, in constant time, if $s$ is or is not a junction of any pair of $G$. Second, we apply this idea for all $s$ in $G$. Therefore, to solve the {\sc $k$-pairs-all-junctions} problem we spend $O(n(m+k))$ time and, or $O(m)$ space, if the junctions are listed during execution, or $O(nk + m)$ space, in the worst case, if they are stored. We have applied our algorithm on four kinship networks brazilian indian ethnic groups: arara, deni, enawen\^e-naw\^es and irantxe-myky. The results and consequences in Anthropology will be analyzed in a forthcoming paper.\\

\noindent{\bf Acknowledgement}

We thank an anonymous referee for many interesting sugestions in a former version of this paper.

\newpage
\noindent\textbf{Appendix}\\

\noindent We present here a proof of the Lemma \ref{l:7}.\\

{\noindent\bf Lemma 6.}
{\it
Suppose a call to ${\proc{Single-Junction-All-Pairs} (s, s_i)}$, where $s_i \in \delta_{T_s}^+(s)$. In line 3, the following assertions are true for the parameter vertex $z$ and for all pairs of vertices $u, v \in T_{s_i}$, such that $post[u], post[v] \geq post[w]+1$. 
\begin {itemize}
\item The vertex $z$ is a representative;
\item The values $p[u]$ and $p[v]$ point to a representative;
\item For all vertices $x \in T_{s_i} \setminus T_z$, $p[z] \neq p[x]$; and
\item  The vertex $s \in J(u,v)$ if, and only if, $p[u] \neq p[v]$.
\end{itemize}
}
\begin{proof}
The first assertion is true as in the first call $z=s_i$ is a representative and before any recursive call to $\proc{Single-Junction-All-Pairs}(s, w)$ we make $p[w] = w$. The second assertion is true as line 5 or line 8 is executed on all $w$ proper descendants of $s_i$ and, in both cases, $p[w]$ points to a representative. The third assertion is true as for any vertex $v \in T_{s}$, $p[v]$ points to a representative that is its ancestor in $T_s$. By first assertion, $p[z] = z$. But the vertex $z$ is not an ancestor of $x$. Therefore, $p[z]\neq p[x]$. We prove the fourth assertion by induction in the number of iterations of the line 3. In the first iteration of line 3 we have $z=s_i$ and $w$ is the child of the greatest post-order value, $post[w]=post[s_i] - 1$. This means that $s_i$ is the single vertex in $T_{s_i}$ with post-order value greater than $post[w]$. By definition, $s \notin J(s_i,s_i)$ and $p[s_i] = p[s_i]$. Now suppose that the fourth assertion is true in the $i$-th iteration of line 3 of a call to $\proc{Single-Junction-All-Pairs}(s, z)$. Let us prove that the assertion is true in the $(i+1)$-th iteration of this same call. Consider the directed path in $T_z$ from $z$ to $w$, $Z = \{z=w_0, \dots, w_{k-1}=w\}$ in the $i$-th iteration. This path was configured during previous iterations and, for its construction, we know that $s \notin J(z, w_l)$, $l=0, \dots, k-2$. By Property \ref{p:1}, $post[z], post[t] \geq post[w] + 1$, for $z$ and for all $t \in \delta_G^-(w)$. Fourth assertion is valid for pair $z, t$, so $s \in J(z, t)$ if, and only if, $p[z] \neq p[t]$. We can apply Lemma \ref{l:3} and we will have two possibilities. \\
{\bf Case 1.} The vertex $s \notin J(z, t)$, for all $t \in \delta_{G}^-(w)$. So, $p[z] = p[t]$, for all $t \in \delta_{G}^-(w)$ and by Lemma \ref{l:3}, $s \notin J(z,w)$. The line 8 is executed and we have $p[w] = z$. Take $x \in T_{s_i}$ with $post[x] \geq post[w]+1$ (taking a vertex $x$ with $post[x] = post[w]$ we have, by definition, $s \notin J(w,w)$ and, of course, $p[w] = p[w]$). Consider $y=LCA_{T_s}(w,x)$ (possibly $x = y$). Let us show that exactly before line 9 be executed, $s \in J(x,w)$ if, and only if, $p[x] \neq p[w]$. Suppose $x \in T_{s_i}\setminus T_z$. Then, the vertex $y$ is some vertex of the directed path from $s_i$ to $z$ in $T_{s_i}$ and, by third assertion, $p[y] \neq p[z]$ and $p[x] \neq p[z] = z = p[w]$. Therefore, we have to show that $s \in J(x, w)$. We have $s \in J(y, z)$ since $p[y] \neq p[z]$. As $G$ is a DAG, we can use the directed path from $z$ to $w$ in $T_{s_i}$ to obtain $s \in J(y,w)$. If $y = x$, then is done. Else, by Lemma \ref{l:2}, $s \in J(x, w)$. Now suppose $x \in T_z$. So, $y = LCA_{T_s}(w, x) = w_j$, $w_j \in Z$ (possibly $x = y$).\\
{\bf Case 1.1.} The vertex $s \notin J(x,z)$. In this case, $p[x] = p[z]$. Moreover, $p[w] = z = p[z] = p[x]$. So, we have to show that $s \notin J(x, w)$. As we have said, $s \notin J(z, w_l)$, for $l = 0, \dots, k-1$. By Lemma \ref{l:5}, in particular $s \notin J(y,w)$. Thus, if $x = y$, then $s \notin J(x, w)$. If $x \neq y$, then, by the execution of the algorithm and Lemma \ref{l:5}, $s$ cannot be a junction of any pairs $x, x_r$, where $x_r$ is a vertex in the directed path from $z$ to $x$ in $T_{z}$. One of these vertices is $y$. Therefore, $s \notin J(y, x)$. Finally, by Lemma 2, $s \notin J(x, w)$ because $s \notin J(y, w)$ and $s \notin J(y, x)$.\\
{\bf Case 1.2.} The vertex $s \in J(x, z)$. Then, $p[x] \neq p[z]$. Moreover, $p[w] = z = p[z] \neq p[x]$. So, we have to show that $s \in J(x,w)$. We know that $s \notin J(z, y)$. Then our assertion says $p[z] = p[y]$. So, $p[y]=p[z]\neq p[x]$. It means that $s\in J(y, x)$. Again, by Lemma 2, $s\in J(x,w)$. After the execution of line 9, our assertion is restored.\\
{\bf Case 2.} The vertex $s \in J(z, t)$, for some $t \in \delta_{G}^-(w)$. So, $p[z] \neq p[t]$, for some $t \in \delta_{G}^-(w)$ and by Lemma \ref{l:3}, $s \in J(z,w)$. Line 5 is executed and we have $p[w] = w$. Take $x \in T_{s_i}$ with $post[x] \geq post[w]+1$ (again, taking a vertex $x$ with $post[x] = post[w]$ we have, by definition, $s \notin J(w,w)$ and, of course, $p[w] = p[w]$). We are going to show that exactly before the execution of line 6, $s \in J(x,w)$ if, and only if $p[x] \neq p[w]$. As $w$ is not an ancestor in $T_s$ of $x$ and $p[x]$ points to a representative that is its ancestor in $T_s$, we have $p[x] \neq p[w]$. Thus, we have to show that $s \in J(x, w)$. If $x \in T_z \setminus T_w$, then Lemma \ref{l:4}.b ensures $s \in J(x,w)$. If $x \in T_{s_i}\setminus T_z$, then $y = LCA_{T_s}(x,w)$ is a vertex in the directed path from $s_i$ to $z$ in $T_{s_i}$ (possibly $x = y$). By third assertion, $p[y] \neq p[z]$. Thus, $s \in J(y, z)$, and we can use the directed path from $z$ to $w$ in $T_{s_i}$ to obtain $s \in J(y,w)$. If $x = y$, then is done. If, $x \neq y$, then, by Lemma \ref{l:2}, $s \in J(x,w)$. After the recursive call execution in line 6, we have for all $u, v \in T_{s_i}$, $post[u],post[v] \geq post(\proc{vertex}(minpost[w]))$, the vertex $s \in J(u, v)$ if, and only if, $p[u] \neq p[v]$. Our assertion is restored after line 7. \qed
\end{proof}
\end{document}